\documentclass[aps, 
twocolumn,
superscriptaddress,
preprintnumbers,amsmath,amssymb]{revtex4}

\usepackage{graphicx}
\usepackage[caption=false]{subfig}
\usepackage{booktabs}
\usepackage{diagbox}
\usepackage{amsthm}
\usepackage{tensor}
\usepackage{color}
\usepackage[all]{xy}
\usepackage{tikz}
\usepackage{dsfont}
\usepackage{times,txfonts}
\usetikzlibrary{positioning}
\usepackage{braket}
\usepackage{kotex}
\usepackage{mathtools}
\usepackage[T1]{fontenc}
\usepackage{enumitem}
\usepackage[toc,page,titletoc]{appendix}
\usepackage{hyperref}
\usepackage{xspace}
\setlist[enumerate]{label=(\roman*),align=left,labelsep=.5em,leftmargin=*,widest=ix,itemsep=0pt,parsep=0pt,topsep=.3\baselineskip}

\newtheorem{Thm}{Theorem}
\newtheorem{Lem}[Thm]{Lemma}
\newtheorem{Prop}[Thm]{Proposition}
\newtheorem{Cor}[Thm]{Corollary}
\theoremstyle{definition}
\newtheorem{Def}[Thm]{Definition}

\newtheorem{Example}[Thm]{Example}

\newtheorem{Conj}[Thm]{Conjecture}

\newcommand{\C}{\mathbb C}
\newcommand{\Tr}{\operatorname{Tr}}
\newcommand{\id}{\operatorname{id}}
\newcommand{\rank}{\operatorname{rank}}
\newcommand{\SR}{\operatorname{SR}}
\newcommand{\SN}{\operatorname{SN}}
\newcommand{\SEP}{\mathsf{SEP}}
\newcommand{\sPPT}{\mathsf{PPT}}
\newcommand{\cCP}{\mathcal{CP}}
\newcommand{\cPPT}{\mathcal{PPT}}

\newcommand{\sS}{\mathsf S}
\newcommand{\sX}{\mathsf X}
\newcommand{\sY}{\mathsf Y}
\newcommand{\sUND}{\mathsf{UND}}
\newcommand{\EB}{\mathrm{EB}}
\newcommand{\cSP}{\mathcal{SP}}
\newcommand{\cUND}{\mathcal{UND}}
\newcommand{\cX}{\mathcal X}
\newcommand{\cY}{\mathcal Y}
\newcommand{\M}{\mathsf M}
\newcommand{\cEB}{\mathcal{EB}}
\newcommand{\ox}{\otimes}
\newcommand{\Id}{\mathbb I}
\newcommand{\proj}[1]{\ket{#1}\!\bra{#1}}

\begin{document}

\title{Beyond the Positive Partial Transpose Squared Conjecture: The Qutrit Case}

\author{Junhyeong An}
\affiliation{Department of Mathematics and Research Institute for Basic Sciences, Kyung Hee University, Seoul 02447, Korea}

\author{Soojoon Lee}
\affiliation{Department of Mathematics and Research Institute for Basic Sciences, Kyung Hee University, Seoul 02447, Korea}
\affiliation{School of Computational Sciences, Korea Institute for Advanced Study, Seoul 02455, Korea}

\date{\today}

\begin{abstract}
Entanglement swapping is a fundamental operation in quantum repeaters for establishing entanglement between distant parties. 
The positive partial transpose (PPT) squared conjecture asks whether two PPT entangled links can generate terminal entanglement through entanglement swapping, or equivalently, whether the composition of two PPT maps is always entanglement breaking. 
Motivated by this conjecture, we investigate the map-composition problem beyond the PPT setting. 
For qutrit completely positive (CP) maps, we prove that the composition of any CP map whose Choi matrix is
$1$-undistillable with any CP map whose Choi matrix has Schmidt number at most two is entanglement breaking
in either order.  
Moreover, we show that the cone of $1$-undistillable CP maps is exactly the largest qutrit cone of CP maps whose composition with every CP map whose Choi matrix has Schmidt number at most two is entanglement breaking in both orders. 
Finally, although map composition captures only the standard maximally entangled outcome in entanglement swapping, we prove that any $1$-undistillable two-qutrit state and any state of Schmidt number at most two cannot generate terminal entanglement under an arbitrary selective measurement on the intermediate systems.
\end{abstract}

\maketitle

\section{Introduction}\label{sec:introduction}

Entanglement is one of the central resources in quantum communication~\cite{nielsen2010quantum}. 
In long-distance quantum communication, the basic task is not merely to create entanglement, but to distribute it reliably between distant parties. 
However, direct transmission of entangled states is fundamentally limited by loss and noise in realistic physical channels. 
Quantum repeaters were introduced to overcome this limitation by enabling entanglement distribution over long distances through shorter elementary links~\cite{briegel1998quantum,duan2001long}.

In the central strategy of a quantum repeater, several shorter entangled links are  connected by entanglement swapping, where a measurement on the intermediate systems attempts to generate entanglement between the terminal systems~\cite{zukowski1993event,bennett1993teleporting}. 
Thus the usefulness of a noisy entangled link is not determined only by whether the state is entangled.
It also depends on whether the link can properly participate in an entanglement swapping step that produces end-to-end entanglement.

A standard method for making noisy entanglement useful is entanglement distillation~\cite{bennett1996purification,deutsch1996quantum}, which extracts fewer highly entangled states from many copies of weakly entangled states. 
Such procedures, however, require multiple copies and quantum memory. 
It is therefore natural to ask whether entanglement can be achieved by a single entanglement swapping step without prior distillation.

This question becomes particularly subtle for positive partial transpose (PPT) entangled states~\cite{peres1996separability,horodecki1996necessary}, since some PPT entangled states have positive key rates in quantum key distribution~\cite{horodecki2005secure,chi2007bound} although they are nondistillable~\cite{horodecki1998mixed,divincenzo2000evidence,dur2000distillability}.
The PPT squared conjecture asks whether two PPT links can generate entanglement between the terminal systems through entanglement swapping~\cite{christandl2012ppt,christandl2019composed}. 
If the conjecture holds, then PPT entangled links cannot generate terminal entanglement, and hence cannot produce any secret key between the end parties in the repeater-based quantum key distribution system using PPT states as entangled links.

Using the Choi--Jamiołkowski correspondence~\cite{choi1975completely,jamiolkowski1972linear}, one obtains the map formulation of the PPT squared conjecture. 
Under this correspondence, PPT matrices correspond to PPT maps, and separable matrices correspond to entanglement breaking maps~\cite{horodecki2003entanglement}. 
The map formulation asks whether the composition of two arbitrary PPT maps is always entanglement breaking. 
In the qutrit case, this conjecture is known to hold, while it still remains open for higher-dimensional cases~\cite{chen2019positive,christandl2019composed}.

In this paper, we go beyond the PPT squared conjecture and consider a more general composition problem for cones of completely positive (CP) maps.
Specifically, we say that an ordered pair $(\cX,\cY)$ of cones of CP maps is \emph{EB-composable} if $\Phi\circ\Psi$ is entanglement breaking for every $\Phi\in\cX$ and $\Psi\in\cY$.
We refer to the problem of determining such ordered pairs as the \emph{EB-composability problem}, of which the PPT squared conjecture is the special case for $(\cX,\cY)=(\cPPT,\cPPT)$, where $\cPPT$ denotes the cone of PPT maps (see Fig.~\ref{fig:1}).

\begin{figure}
	\centering
	\includegraphics[width=1.0\linewidth]{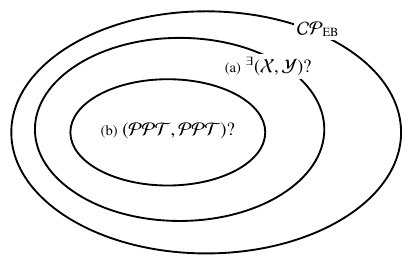}
	\caption{Schematic illustration of the EB-composability problem.
Let $\cCP_{\EB}$ be the set of ordered pairs
$(\Phi,\Psi)$ of CP maps such that $\Phi\circ\Psi$ is entanglement breaking.
(a)  An ordered pair of cones $(\cX,\cY)$ is EB-composable precisely when $\cCP_{\EB}$ contains the Cartesian product of the two cones $\cX\times\cY$; in the figure, we identify the ordered pair $(\cX,\cY)$ with this product.
(b) In particular, if, in every dimension, an EB-composable pair $(\cX,\cY)$ satisfies $\cPPT\subseteq\cX$ and $\cPPT\subseteq\cY$, then the PPT squared conjecture holds.
}
	\label{fig:1}
\end{figure}

In the two-qutrit setting, we exhibit EB-composable pairs beyond PPT.
We prove that both $(\cUND_1,\cSP_2)$ and $(\cSP_2,\cUND_1)$ are EB-composable, where $\cUND_1$ denotes the cone of CP maps whose Choi matrices are $1$-undistillable~\cite{horodecki1998mixed,dur2000distillability}, and $\cSP_2$ denotes the cone of CP maps whose Choi matrices have Schmidt number at most two~\cite{terhal2000schmidt}  (see Fig.~\ref{fig:2}).
Moreover, we show that when one is fixed to be $\cSP_2$, the largest compatible cone of CP maps is exactly $\cUND_1$, regardless of the order of composition.

\begin{figure}
	\centering
	\includegraphics[width=1.0\linewidth]{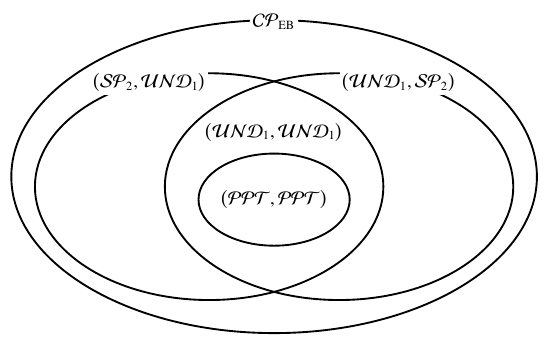}
	\caption{Schematic illustration of our EB-composability results in the qutrit setting.
For qutrit CP maps, $\cPPT\subseteq\cUND_1\subseteq\cSP_2$~\cite{chen2017schmidt}, and both $\cUND_1\times\cSP_2$ and $\cSP_2\times\cUND_1$ are contained in $\cCP_\EB$ (see Theorem~\ref{thm:und-sp2-eb-composable}).
Consequently, $(\cUND_1\times\cSP_2)\cap(\cSP_2\times\cUND_1)
=\cUND_1\times\cUND_1\subseteq\cCP_{\EB}$ (Corollary~\ref{cor:und-und-eb}).
Since $\cUND_1$ contains non-PPT maps, 
$\cPPT\times\cPPT\subsetneq\cUND_1\times\cUND_1$.
Moreover, when either component is fixed to be $\cSP_2$, the largest compatible cone for the other component is exactly $\cUND_1$ (see Theorem~\ref{thm:maximality-und1-sp2}).
}
	\label{fig:2}
\end{figure}

We next explain the meaning of this EB-composability result in the repeater picture. 
Map composition corresponds only to the standard maximally entangled post-selected outcome for Choi matrices in entanglement swapping, whereas a general selective repeater protocol allows arbitrary positive operators on the intermediate systems. 
We show that, for cones of bipartite positive semidefinite matrices that are closed under local filtering, EB-composability of the cones of the corresponding CP maps implies that any pair of states in the cones of matrices cannot generate entanglement between the two end systems through any post-selected entanglement swapping step.

We then apply this result to the two-qutrit state cones $\sUND_1$ and $\sS_2$ which are ones of 1-undistillable states and states with Schmidt number at most two, respectively.
By showing that these cones are closed under local filtering, we can prove that any two-qutrit state in $\sUND_1$ and any two-qutrit state in $\sS_2$ cannot produce end-to-end entanglement by any selective measurement on the intermediate systems, since the cones of the corresponding CP maps are EB-composable in both orders by our above results.
In particular, even when both elementary links are non-PPT entangled states in $\sS_2$ and one of them is distillable, their direct use in entanglement swapping without prior processing need not generate end-to-end entanglement, showing that distillability and non-PPT entanglement do not guarantee direct usefulness as a repeater link.

The remainder of this paper is organized as follows.
In Sec.~\ref{sec:channel-composition}, we define the EB-composability problem and prove the qutrit composition and maximality results.
In Sec.~\ref{sec:swapping}, we clarify the relation between map composition and entanglement swapping, and show how the map composition conclusion extends to entanglement swapping. 
Finally, in Sec.~\ref{sec:discussion}, we summarize the implications of our results and discuss remaining questions.

\section{Map composition problem beyond the PPT squared conjecture}
\label{sec:channel-composition}

Throughout this paper, we consider bipartite systems on $d\otimes d$, with particular focus on the two-qutrit case $d=3$.
To keep the main text focused on the main ideas, we use the standard notation and definitions introduced in Appendix~\ref{appen:preli}.

The PPT squared conjecture is motivated by entanglement swapping.
In its state-level formulation, one starts with two PPT states and asks whether a selective measurement can generate entanglement between the remaining systems.
The conjecture asserts that this never happens.
This gives the following state-level formulation.

\begin{Conj}[PPT squared conjecture in the state formulation]
\label{conj:ppt-state}
Let $\rho_{AB}$ and $\sigma_{CD}$ be PPT states on $\mathbb{C}^d\otimes\mathbb{C}^d$.
Then, for every positive operator $M_{BC}\ge0$, the post-selected unnormalized state
\[
\Tr_{BC}
[
(\rho_{AB}\otimes\sigma_{CD})
(I_A\otimes M_{BC}\otimes I_D)
]
\]
is separable on $AD$.
\end{Conj}

We now isolate the part of this problem that naturally leads to map composition.
Let $\ket{\phi_d^+}=\sum_{i=0}^{d-1}\ket{ii}$ be the (unnormalized) standard maximally entangled state.
For a linear map $\Phi:\M_d\to\M_d$, we call the matrix 
\begin{equation}
\label{eq:Choi}
C_\Phi
=
(\id_d\otimes\Phi)(\proj{\phi_d^+})
\end{equation}
the Choi matrix associated with $\Phi$.
If $C_\Psi$ is placed on $AB$ and $C_\Phi$ on $CD$, then post-selecting the standard maximally entangled outcome on $BC$ gives
\begin{equation}
\label{eq:map_composition}
\Tr_{BC}
[
(C_\Psi^{AB}\otimes C_\Phi^{CD})
(I_A\otimes\proj{\phi_d^+}_{BC}\otimes I_D)
]
=
C_{\Phi\circ\Psi}^{AD}.
\end{equation}
Consequently, separability of this post-selected state is equivalent to the composed map $\Phi\circ\Psi$ being entanglement breaking.

Under the Choi--Jamiołkowski correspondence, PPT states correspond to PPT maps, while separable states correspond to entanglement breaking maps.
Restricting the state formulation to the maximally entangled post-selection therefore gives a map formulation of the conjecture.
This formulation disregards arbitrary measurement outcomes and retains only the outcome in Eq.~\eqref{eq:map_composition}.
In this sense, the PPT squared conjecture leads to the following map composition problem.

\begin{Conj}[PPT squared conjecture in the map formulation]
\label{conj:ppt-channel}
For every pair of PPT maps $\Phi,\Psi:\M_d\to\M_d$, one has $\Phi\circ\Psi\in\cEB$. 
Equivalently, $\cPPT\circ\cPPT=\{\Phi\circ\Psi:\Phi,\Psi\in\cPPT\}\subseteq\cEB$, where $\cEB$ and $\cPPT$ denote the cones of entanglement breaking and PPT maps, respectively. 
\end{Conj}

Motivated by this formulation, we study a more general map composition problem.
The point is to replace the pair $(\cPPT,\cPPT)$  by other natural cones of CP maps and ask whether their compositions are always entanglement breaking.
Since the order of composition matters in general, we formulate the problem for ordered pairs.
This leads to the following definition.

\begin{Def}\label{def:EB-composable}
Let $\cX$ and $\cY$ be cones of CP maps from $\M_d$ to itself.
We say that the ordered pair $(\cX,\cY)$ is \emph{EB-composable} if $\cX\circ\cY\subseteq\cEB$, where $\cX\circ\cY=\{\Phi\circ\Psi: \Phi\in\cX, \Psi\in\cY\}$.
In other words, this means that $\Phi\circ\Psi$ is entanglement breaking for every $\Phi\in\cX$ and every $\Psi\in\cY$.
\end{Def}

The map formulation of the PPT squared conjecture asserts that $(\cPPT,\cPPT)$ is EB-composable.
Our main result concerns two different standard cones of CP maps, namely the cone $\cSP_2$ of $2$-superpositive maps and the cone $\cUND_1$ of $1$-undistillable maps.
The definitions and elementary properties of these cones are recalled in Appendix~\ref{appen:preli}.
Our first main result is the following two-qutrit composition theorem beyond the PPT.

\begin{Thm}
\label{thm:und-sp2-eb-composable}
In the qutrit setting, both $(\cUND_1,\cSP_2)$ and $(\cSP_2,\cUND_1)$ are EB-composable.
\end{Thm}

\begin{proof}
See Appendix~\ref{appen:pfthm_und-sp2-eb-composable}.
\end{proof}

The two conclusions of Theorem~\ref{thm:und-sp2-eb-composable} correspond to the two possible orders of composition.
This is not a formal symmetry of EB-composability, since in general $\Phi\circ\Psi\in\cEB$ does not imply $\Psi\circ\Phi\in\cEB$.
Thus the theorem asserts a genuine two-sided property of the pair $(\cUND_1,\cSP_2)$, not merely a consequence of one order.
A simple example showing the failure of order symmetry in general is given in Appendix~\ref{app:order}.

This result should be compared with the PPT squared conjecture, which asks whether $(\cPPT,\cPPT)$ is EB-composable.
By contrast, Theorem~\ref{thm:und-sp2-eb-composable} concerns the mixed pair $(\cUND_1,\cSP_2)$ in the qutrit setting.
In this sense, it identifies a composition phenomenon of the same type beyond the PPT cone.

The result should not be interpreted as a converse statement.
The fact that $\Phi^2$ is entanglement breaking does not imply that $\Phi$ is $1$-undistillable.
The following example records this point.
It also motivates the maximality theorem below.

\begin{Example}
\label{ex:WH-square}
For $d\ge 3$, consider the Werner--Holevo channel
\[
\Lambda_{\mathrm{WH}}(X)=\frac1{d-1}(\Tr(X)I_d-X^T)
=
\frac1{d-1}\sum_{0\le i<j\le d-1}F_{ij}XF_{ij}^{\dagger},
\]
where $F_{ij}:=\ket{i}\bra{j}-\ket{j}\bra{i}$.
Since all Kraus operators $F_{ij}$ in this representation have rank at most two and the Choi matrix of $\Lambda_{\rm WH}$ is $1$-distillable~\cite{divincenzo2000evidence}, $\Lambda_{\mathrm{WH}}\in\cSP_2\setminus\cUND_1$.
Nevertheless, $\Lambda_{\mathrm{WH}}^2$ is entanglement breaking, since
$\Lambda_{\rm WH}^2(X)= \frac{1}{(d-1)^2}X+\frac{d-2}{(d-1)^2}\Tr(X)I_d$
is the depolarizing channel whose normalized Choi matrix is the isotropic state with $F:=\frac1{d^2}
\langle\phi_d^+|C_{\Lambda_{\mathrm{WH}}^2}|\phi_d^+\rangle
=\frac2{d(d-1)}\le\frac1d$, and hence it is  separable~\cite{horodecki1999reduction}.
Thus $\Phi^2\in\cEB$ does not imply $\Phi\in\cUND_1$.
\end{Example}

Example~\ref{ex:WH-square} raises a natural maximality question.
If the testing cone is fixed to be $\cSP_2$, one can ask how large the opposite cone can be while preserving EB-composability.
The next theorem shows that the answer is exactly $\cUND_1$.
Thus Theorem~\ref{thm:und-sp2-eb-composable} is not only a positive composition result, but also a maximality statement.

\begin{Thm}
\label{thm:maximality-und1-sp2}
Let $\cX$ be a cone of CP maps from \(\M_3\) to itself. 
Then the following are equivalent:
\begin{enumerate}
\item $\cX\subseteq\cUND_1$.
\item $(\cX,\cSP_2)$ is EB-composable.
\item $(\cSP_2,\cX)$ is EB-composable.
\end{enumerate}
Consequently, $\cUND_1$ is the largest cone of qutrit CP maps that is EB-composable with $\cSP_2$, on either side.
\end{Thm}

\begin{proof}
See Appendix~\ref{appen:pfthm_maximality-und1-sp2}.
\end{proof}

The maximality statement is relative to the full testing cone $\cSP_2$.
More precisely, it characterizes the qutrit CP maps whose composition with every map in $\cSP_2$ is entanglement breaking, in either order.
Thus, $\cUND_1$ is the largest cone that is EB-composable, on either side, with all maps whose Choi matrices have Schmidt number at most two.

Equivalently, the contrapositive shows that, for every $\Phi\notin\cUND_1$, there exist $\Psi_{\mathrm L},\Psi_{\mathrm R}\in\cSP_2$ such that $\Psi_{\mathrm L}\circ\Phi$ and $\Phi\circ\Psi_{\mathrm R}$ are not entanglement breaking, where the two testing maps need not coincide.
Thus, compositions with maps in $\cSP_2$ provide a way to determine whether a given CP map belongs to $\cUND_1$.
For a smaller testing cone $\cY\subsetneq\cSP_2$, however, some maps outside $\cUND_1$ may remain undetected, so the maximal compatible cone may be larger, as seen in Example~\ref{ex:WH-square}.

Theorem~\ref{thm:maximality-und1-sp2} gives the following square-type consequence.
We use the known two-qutrit fact that every $1$-undistillable two-qutrit state has Schmidt number at most two~\cite{chen2017schmidt}.
This says that $\cUND_1\subseteq\cSP_2$ for qutrit CP maps.
Combining this inclusion with Theorem~\ref{thm:und-sp2-eb-composable} gives the following result.

\begin{Cor}
\label{cor:und-und-eb}
For qutrit CP maps, $(\cUND_1,\cUND_1)$ is EB-composable.
\end{Cor}

\section{Repeater interpretation beyond map formulation}
\label{sec:swapping}

We now return to the operational state-level picture.
A repeater protocol may post-select on an arbitrary positive operator on $BC$, whereas map composition corresponds, for Choi matrices, only to the special post-selection $M_{BC}=|\phi_d^+\rangle\langle\phi_d^+|_{BC}$.
Thus one may wonder whether separability of the maximally entangled outcome,
\[
\Tr_{BC}[(\rho_{AB}\otimes\sigma_{CD})(I_A\otimes|\phi_d^+\rangle\langle\phi_d^+|_{BC}\otimes I_D)]\in\SEP(A:D),
\]
where $\SEP(A:D)$ denotes the cone of separable positive semidefinite operators between the systems $A$ and $D$, already implies separability for every positive operator $M_{BC}\ge0$:
\[
\Tr_{BC}[(\rho_{AB}\otimes\sigma_{CD})(I_A\otimes M_{BC}\otimes I_D)]\in\SEP(A:D).
\]
If this holds, then no selective measurement outcome on $BC$, whenever it occurs with positive probability, can generate entanglement between $A$ and $D$.

The following example shows that such an interpretation can be false.

\begin{Example}
\label{ex:postselected-outcome-not-enough}
Consider the two-qutrit pure states
\[
\ket{\psi_1}_{AB}
=
\frac{1}{\sqrt2}(\ket{00}+\ket{11}),
\quad
\ket{\psi_2}_{CD}
=
\frac{1}{\sqrt2}(\ket{00}+\ket{22}).
\]
For the post-selected $M_{BC}=\proj{\phi_3^+}_{BC}$ outcome, where
$\ket{\phi_3^+}_{BC} =\sum_{i=0}^2 \ket{ii}_{BC},$
we obtain
\[
{}_{BC}\!\bra{\phi_3^+}
(\ket{\psi_1}_{AB}\otimes\ket{\psi_2}_{CD})
=
\frac{1}{2}\ket{00}_{AD},
\]
which is a product vector.
Hence this post-selected outcome is separable.

On the other hand, let $\sigma_X\ket j=\ket{j+1\pmod 3}$, and consider another unnormalized Bell state
\[
\ket{\eta}_{BC}
=
(\sigma_X\otimes I)\ket{\phi_3^+}_{BC}.
\]
Since $\sigma_X^\dagger=\sigma_X^2$, we have
\[
\begin{aligned}
{}_{BC}\!\bra{\eta}
(\ket{\psi_1}_{AB}\otimes\ket{\psi_2}_{CD})
&=
{}_{BC}\!\bra{\phi_3^+}(\sigma_X^2\otimes I)
(\ket{\psi_1}_{AB}\otimes\ket{\psi_2}_{CD}) \\
&=
\frac{1}{2}
(\ket{02}_{AD}+\ket{10}_{AD}),
\end{aligned}
\]
which is entangled.
Thus separability of the post-selected outcome by $\ket{\phi_3^+}$ does not imply separability of every selective entanglement swapping outcome.
\end{Example}

The next result gives a condition under which the post-selected outcome by $\ket{\phi_d^+}$ is nevertheless sufficient to control all selective outcomes.
Here, being closed under local filtering means that whenever $\rho_{AB}\in\sX$, the operator $(L_A\otimes R_B)\rho_{AB}(L_A^\dagger\otimes R_B^\dagger)$ also belongs to $\sX$ for all local operators $L_A,R_B$, whenever it is nonzero.
We impose the analogous condition on $\sY$.
In this way, the cones are closed under the local changes produced by different measurement effects.

\begin{Thm}
\label{thm:closed-under-filtering-reduction}
Let $\sX$ and $\sY$ be cones of bipartite positive semidefinite matrices on $\C^d\otimes\C^d$. 
Suppose that both cones are closed under local filtering.
Assume that for every $\rho_{AB}\in\sX$ and every $\sigma_{CD}\in\sY$,
\[
\Tr_{BC}
[
(\rho_{AB}\otimes\sigma_{CD})
(I_A\otimes\proj{\phi_d^+}_{BC}\otimes I_D)
]
\in\SEP(A:D).
\]
Then, for every $\rho_{AB}\in\sX$, every $\sigma_{CD}\in\sY$, and every positive operator $M_{BC}\ge0$,
\[
\Tr_{BC}
[
(\rho_{AB}\otimes\sigma_{CD})
(I_A\otimes M_{BC}\otimes I_D)
]
\in\SEP(A:D).
\]
\end{Thm}

\begin{proof}
See Appendix~\ref{appen:entanglement swapping}.
\end{proof}

We first note that the cone of PPT matrices, denoted by $\sPPT$, is closed under local filtering. 
Hence, when $\sX=\sY=\sPPT$, Theorem~\ref{thm:closed-under-filtering-reduction} shows that the maximally entangled post-selection formulation and the arbitrary selective-measurement formulation of the PPT squared conjecture are equivalent.
The proof shows that this passage is not specific to PPT states; the essential assumption is being closed under local filtering. 
Consequently, the same argument applies to any pair of cones
closed under local filtering.

We now apply Theorem~\ref{thm:closed-under-filtering-reduction} to $\sUND_1$ and $\sS_2$, which are defined as the cones of all $1$-undistillable positive semidefinite matrices and all positive semidefinite matrices with Schmidt number at most two on $\mathbb{C}^d\otimes\mathbb{C}^d$, respectively.

By the Choi--Jamiołkowski correspondence, Theorem~\ref{thm:und-sp2-eb-composable} implies that, for every
$\rho_{AB}\in\sUND_1$ and $\sigma_{CD}\in\sS_2$, in either order, 
\[
\Tr_{BC}
[
(\rho_{AB}\otimes\sigma_{CD})
(I_A\otimes\proj{\phi_3^+}_{BC}\otimes I_D)
]
\in \SEP(A:D).
\]
Since \(\sUND_1\) and \(\sS_2\) are closed under local filtering, as recalled in Appendix~\ref{appen:preli}, Theorem~\ref{thm:closed-under-filtering-reduction} extends this conclusion from the maximally entangled effect to every positive effect \(M_{BC}\ge0\).
This gives the state-level entanglement swapping consequence.

\begin{Cor}
\label{cor:und-s2-arbitrary-swapping}
Let $\rho_{AB}$ and $\sigma_{CD}$ be two-qutrit states.
Suppose that
\[
\rho_{AB}\in\sUND_1\text{ and } \sigma_{CD}\in\sS_2, 
\text{ or }
\rho_{AB}\in\sS_2\text{ and }  \sigma_{CD}\in\sUND_1.
\]
Then, for every positive operator $M_{BC}\ge0$,
\[
\Tr_{BC}
[
(\rho_{AB}\otimes\sigma_{CD})
(I_A\otimes M_{BC}\otimes I_D)
]
\in\SEP(A:D).
\]
\end{Cor}

\begin{figure}
	\centering
	\includegraphics[width=1.0\linewidth]{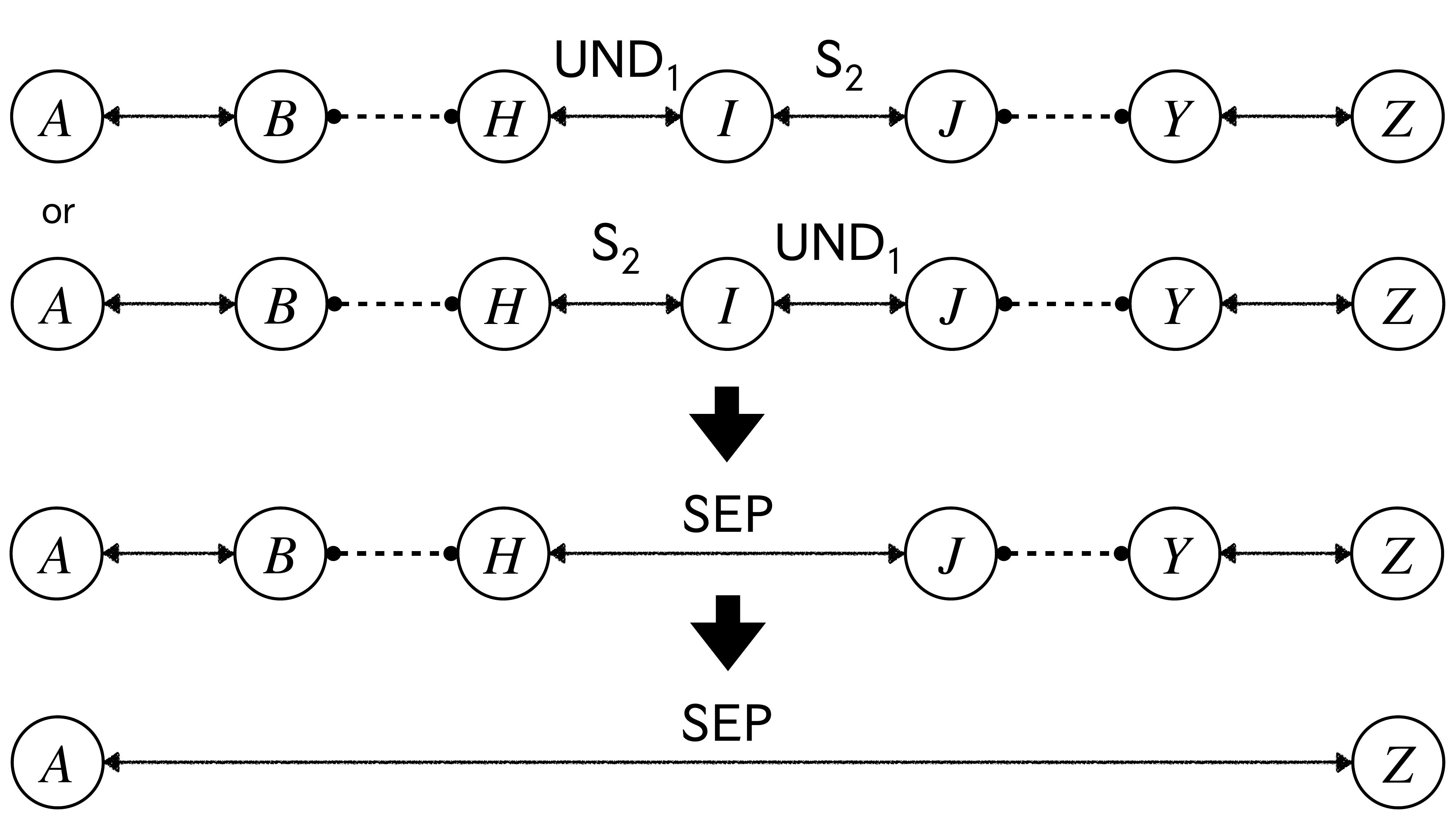}
	\caption{
Entanglement-swapping obstruction caused by a $1$-undistillable link.
When a two-qutrit link in $\sUND_1$ is joined to a link in $\sS_2$, in either order, every post-selected entanglement swapping outcome is separable.
Consequently, the resulting state is separable, and subsequent entanglement swapping cannot generate end-to-end entanglement between $A$ and $Z$.
}
	\label{Fig:repeater}
\end{figure}

Corollary~\ref{cor:und-s2-arbitrary-swapping} gives the entanglement swapping interpretation of Theorem~\ref{thm:und-sp2-eb-composable}.
The universal obstruction in this result arises not from the Schmidt number two condition alone, but from the $1$-undistillability of one link: regardless of the ordering, such a link makes every selective swapping outcome with any link in $\sS_2$ separable, as illustrated in Fig.~\ref{Fig:repeater}.
By contrast, Example~\ref{ex:postselected-outcome-not-enough} shows that membership of both links in $\sS_2$ alone does not determine whether selective entanglement swapping succeeds or fails.

Together with Theorem~\ref{thm:maximality-und1-sp2}, this obstruction is sharp with respect to the testing class $\sS_2$.
Indeed, for every $1$-distillable state, there exists, in each ordering, a suitable partner in $\sS_2$ such that the standard maximally entangled post-selection produces an entangled terminal state, where the partners for the two orderings need not coincide.
This does not mean that entanglement swapping succeeds with every partner in $\sS_2$: even a distillable link in $\sS_2$ cannot generate direct end-to-end entanglement when paired with a $1$-undistillable link, showing that distillability and non-PPT entanglement do not guarantee direct usefulness as a repeater link without prior processing.

\section{Discussion}\label{sec:discussion}

This paper extends the PPT squared problem to the EB-composability problem for cones of CP maps.
The known qutrit PPT squared theorem establishes EB-composability for the pair $(\cPPT,\cPPT)$.
We prove that the same entanglement breaking conclusion holds for $(\cUND_1,\cSP_2)$ and $(\cSP_2,\cUND_1)$. 
Thus $1$-undistillability and Schmidt number two provide a structural obstruction to entanglement generation under map composition.

In addition, our maximality theorem shows that this composition result cannot be enlarged on the $\cUND_1$ side when the testing cone is fixed to be $\cSP_2$.
Indeed, any qutrit CP map outside $\cUND_1$ fails to be EB-composable with $\cSP_2$ in at least one order.
Thus $\cUND_1$ is not merely a sufficient cone for the proof, but the maximal cone compatible with EB-composability against all $2$-superpositive maps in the qutrit system.

For the state-level interpretation, one must distinguish map composition from general entanglement swapping. 
Map composition represents only one maximally entangled post-selected outcome for Choi matrices, while a selective repeater protocol allows arbitrary positive operators on the intermediate systems. 
We show that being closed under local filtering bridges this gap for the cones considered here. 
Consequently, if $\rho_{AB}\in\sUND_1$ and $\sigma_{CD}\in \sS_2$ are two-qutrit states, then no selective measurement on $BC$ can generate entanglement between $A$ and $D$.

This result should not be understood as a characterization of all CP maps whose square is entanglement breaking. 
There exist CP maps whose Choi matrices are $1$-distillable while their squares are nevertheless entanglement breaking. 
Thus the condition $\Phi^2\in\cEB$ can be weaker than the condition $C_\Phi\in\sUND_1$, as seen in Example~\ref{ex:WH-square}, while
the maximality result in this paper concerns the sharp boundary of EB-composability against the testing cone $\cSP_2$.

The present proof is essentially specific to the two-qutrit setting. 
It relies on the fact that PPT implies separability on $2\otimes3$ and $3\otimes2$ systems~\cite{peres1996separability,horodecki1996necessary}. 
In higher dimensions, the same argument only gives PPT operators supported on $2\otimes d$ subspaces, which need not be separable. A
natural next step is to determine which testing cones beyond $\cSP_2$, or which additional assumptions beyond $1$-undistillability, lead to analogous EB-composability results.

Another natural direction is to compare EB-composability with eventual entanglement breaking behavior. 
EB-composability concerns the composition of maps from two prescribed cones, whereas eventual entanglement breaking asks when repeated composition of a single map becomes entanglement breaking~\cite{lami2015entanglement,rahaman2018eventually}. 
These two notions capture different ways in which entanglement can disappear under map composition.
Understanding their relation may clarify how distillability, Schmidt number, and iteration jointly control the loss of entanglement.

\begin{acknowledgments}
We are grateful to Alexander Streltsov and the participants of the seminar hosted at his institute for valuable comments and insightful discussions. 
This research was supported by
the Institute of Information \& Communications Technology Planning \& Evaluation (IITP) grant funded by the Ministry of Science and ICT (MSIT) (No. RS-2025-02304540).
S.L. acknowledges support from the National Research Foundation of Korea (NRF) of Korea grants funded by the MSIT  (No. RS-2024-00432214 and No. RS-2022-NR068791) and Creation of the Quantum Information Science R\&D Ecosystem (No. RS-2023-NR068116) through the NRF funded by the MSIT.
\end{acknowledgments}

\bibliography{EB}

\appendix

\section{Preliminaries}\label{appen:preli}

This appendix collects the notations and elementary facts used in the main text.
We first recall Schmidt rank, Schmidt number, and superpositive maps.
We then present equivalent formulations of $1$-undistillability that are used in the proof of Theorem~\ref{thm:und-sp2-eb-composable}.
Finally, we note that the relevant state cones are closed under local filtering.

A linear map \(\Phi:\M_d\to \M_d\) is positive if \(\Phi(M)\ge0\) for every \(M\ge0\). 
It is CP if
\[
\Id_k\otimes\Phi:\M_k\otimes \M_d\to \M_k\otimes \M_d
\]
is positive for every \(k\ge1\). 

For a vector $\ket{\psi}\in\C^d\otimes\C^d$, its Schmidt rank is denoted by $\SR(\ket{\psi})$, and since
\begin{equation}
\label{eq:SR}
\ket{\psi}=(X\otimes I_d)\ket{\phi_d^+},
\end{equation}
for some matrix $X$, we define $\SR(\ket{\psi})=\rank X$.

The Schmidt number of a positive semidefinite matrix $\rho\in \M_d\otimes \M_d$ is defined by~\cite{terhal2000schmidt}
\[
\SN(\rho)
=
\min_{\rho=\sum_i \proj{\psi_i}}
\max_i \SR(\ket{\psi_i}),
\]
where the minimum is taken over all decompositions of $\rho$ into rank-one matrices.
We denote by $\sS_n$ the convex cone of positive semidefinite matrices with Schmidt number at most $n$.
We denote by $\SEP$ the cone of separable positive semidefinite operators.
Thus $\sS_1$ coincides with $\SEP$.

$\Phi$ is called entanglement breaking if and only if $C_\Phi\in\SEP$.
More generally, $\Phi$ is called $n$-superpositive~\cite{skowronek2009cones}
(also known as an $n$-partially entangled breaking map~\cite{chruscinski2006partially}) if and only if it admits a Kraus representation whose Kraus operators all have rank at most $n$, or equivalently, $C_\Phi\in\sS_n$.
We denote the cone of $n$-superpositive maps by $\cSP_n$.

A CP map $\Phi:\M_d\to \M_d$ is called a PPT map if its Choi matrix is PPT, i.e., $C_\Phi^\Gamma = (\Id\otimes T)C_\Phi\ge0,$ where $T$ is the transpose map.
Equivalently, $\Phi$ is PPT if and only if $T\circ\Phi$ is CP.

We next recall distillability.
For $n\ge1$, a positive semidefinite matrix $\rho\in\M_d\ox\M_d$ is called $n$-distillable if there exists a vector $\ket\psi\in(\C^d)^{\ox n}\ox(\C^d)^{\ox n}$ with $\SR(\ket\psi)\le2$ such that
\[
\bra\psi(\rho^{\ox n})^\Gamma\ket\psi<0.
\]
Otherwise, $\rho$ is called $n$-undistillable~\cite{horodecki1998mixed,divincenzo2000evidence,dur2000distillability}.
We denote by $\sUND_n$ the cone of $n$-undistillable positive semidefinite matrices.
In particular, $\sUND_1$ is a convex cone.
A map $\Phi$ is called $n$-undistillable if its Choi matrix $C_\Phi$ belongs to $\sUND_n$, and we denote the corresponding cone of maps by $\cUND_n$.
In particular, $\cUND_1$ is a convex cone of CP maps.

The following proposition is about the equivalent forms of $1$-undistillability used in the main proofs.
It shows that the condition $\Phi\in\cUND_1$ can equivalently be expressed as follows: for every $\Psi\in\cSP_2$, both $\Phi\circ\Psi$ and $\Psi\circ\Phi$ belong to $\cPPT$.
This is the bridge between $1$-undistillability and PPT conditions for compositions with maps in $\cSP_2$.

\begin{Prop}\label{prop:compression}
Let $\Phi:\M_d\to \M_d$ be a CP map, and let $\rho=C_\Phi$ be its Choi matrix. 
Then the following are equivalent:
\begin{enumerate}
    \item $\Phi\in\cUND_1$.
    \item $(X\ox I_d)\rho(X^\dagger\ox I_d)$ is PPT for every matrix $X\in\M_d$ with $\rank X\le2$.
    \item $(I_d\ox Y)\rho(I_d\ox Y^\dagger)$ is PPT for every matrix $Y\in\M_d$ with $\rank Y\le2$.
    \item For every pure state $\ket{\psi}\in\C^d\ox\C^d$ with $\SR(\ket{\psi})\le2$, $(\Id_d\otimes \Phi)(\proj{\psi})$ is PPT.
    \item For every $\sigma\in\sS_2$, $(\Id_d\otimes \Phi)(\sigma)$ is PPT.
    \item For every $\Psi\in\cSP_2$, $(\Id_d\otimes \Psi)(\rho)$ is PPT.
\end{enumerate}
\end{Prop}

\begin{proof}
We first prove the equivalence between (i) and (ii).
Suppose that \(C_\Phi\in\sUND_1\), let \(X\in\M_d\) have rank at most two, and set
\[
\tau_X=(X\otimes I_d)C_\Phi(X^\dagger\otimes I_d).
\]
For every \(\ket{\eta}\in\C^d\otimes\C^d\),
\[
\bra{\eta}\tau_X^\Gamma\ket{\eta}
=
\bra{\xi}C_\Phi^\Gamma\ket{\xi},
\quad
\ket{\xi}=(X^\dagger\otimes I_d)\ket{\eta}.
\]
Since \(\rank X\le2\), we have \(\SR(\ket{\xi})\le2\).
The \(1\)-undistillability of \(C_\Phi\) gives
\[
\bra{\xi}C_\Phi^\Gamma\ket{\xi}\ge0.
\]
Hence \(\tau_X^\Gamma\ge0\), so \(\tau_X\) is PPT.

Conversely, suppose that
\[
(X\otimes I_d)C_\Phi(X^\dagger\otimes I_d)
\]
is PPT for every \(X\in\M_d\) with \(\rank X\le2\).
Let \(\ket{\xi}\) be an arbitrary vector with \(\SR(\ket{\xi})\le2\).
Then there exists \(X\in\M_d\) with \(\rank X\le2\) such that
\[
\ket{\xi}=(X^\dagger\otimes I_d)\ket{\phi_d^+}.
\]
For
\[
\tau_X=(X\otimes I_d)C_\Phi(X^\dagger\otimes I_d),
\]
the PPT assumption gives
\[
0\le
\bra{\phi_d^+}\tau_X^\Gamma\ket{\phi_d^+}
=
\bra{\xi}C_\Phi^\Gamma\ket{\xi}.
\]
Since this holds for every vector \(\ket{\xi}\) with Schmidt rank at most two,
\(C_\Phi\) is \(1\)-undistillable. This proves \((i)\Leftrightarrow(ii)\).

We next prove the equivalence between (i) and (iii).
Suppose that \(C_\Phi\in\sUND_1\), let \(Y\in\M_d\) have rank at most two, and set
\[
\omega_Y=(I_d\otimes Y)C_\Phi(I_d\otimes Y^\dagger).
\]
For every \(\ket{\eta}\in\C^d\otimes\C^d\),
\[
\bra{\eta}\omega_Y^\Gamma\ket{\eta}
=
\bra{\zeta}C_\Phi^\Gamma\ket{\zeta},
\quad
\ket{\zeta}=(I_d\otimes Y^T)\ket{\eta}.
\]
Since \(\rank Y\le2\), we have \(\SR(\ket{\zeta})\le2\).
Thus
\[
\bra{\eta}\omega_Y^\Gamma\ket{\eta}\ge0
\]
for every \(\ket{\eta}\), and hence \(\omega_Y\) is PPT.

Conversely, suppose that
\[
(I_d\otimes Y)C_\Phi(I_d\otimes Y^\dagger)
\]
is PPT for every \(Y\in\M_d\) with \(\rank Y\le2\).
Let \(\ket{\zeta}\) be an arbitrary vector with \(\SR(\ket{\zeta})\le2\).
Then there exists \(Y\in\M_d\) with \(\rank Y\le2\) such that
\[
\ket{\zeta}=(I_d\otimes Y^T)\ket{\phi_d^+}.
\]
For
\[
\omega_Y=(I_d\otimes Y)C_\Phi(I_d\otimes Y^\dagger),
\]
the PPT assumption gives
\[
0\le
\bra{\phi_d^+}\omega_Y^\Gamma\ket{\phi_d^+}
=
\bra{\zeta}C_\Phi^\Gamma\ket{\zeta}.
\]
Since this holds for every vector \(\ket{\zeta}\) with Schmidt rank at most two,
\(C_\Phi\) is \(1\)-undistillable. This proves \((i)\Leftrightarrow(iii)\).

We now prove \((ii)\Leftrightarrow(iv)\).
Let \(\ket{\psi}\in\C^d\otimes\C^d\) be a vector with \(\SR(\ket{\psi})\le2\).
Then there exists \(X\in\M_d\) with \(\rank X\le2\) such that
\[
\ket{\psi}=(X\otimes I_d)\ket{\phi_d^+}.
\]
Using the definition of the Choi matrix, we obtain
\[
(\Id_d\otimes\Phi)(\proj{\psi})
=
(X\otimes I_d)C_\Phi(X^\dagger\otimes I_d).
\]
Therefore \((\Id_d\otimes\Phi)(\proj{\psi})\) is PPT for every
\(\ket{\psi}\) with Schmidt rank at most two if and only if every
rank-at-most-two local filtering of \(C_\Phi\) on the first subsystem is PPT.
This proves \((ii)\Leftrightarrow(iv)\).

Next, we prove \((iv)\Leftrightarrow(v)\).
Assume that $(iv)$ holds and let \(\sigma\in\sS_2\). Then \(\sigma\) admits a finite decomposition
\[
\sigma=\sum_i\proj{\psi_i},
\quad
\SR(\ket{\psi_i})\le2
\]
for all \(i\). Hence
\[
(\Id_d\otimes\Phi)(\sigma)
=
\sum_i(\Id_d\otimes\Phi)(\proj{\psi_i}).
\]
Each summand is PPT by the assumption, and therefore
\((\Id_d\otimes\Phi)(\sigma)\) is PPT.
The converse is immediate by taking
\(\sigma=\proj{\psi}\) with \(\SR(\ket{\psi})\le2\).
Thus \((iv)\Leftrightarrow(v)\).

Finally, we prove \((iii)\Leftrightarrow(vi)\).
Assume first that $(iii)$ holds and let \(\Psi\in\cSP_2\).
Then \(\Psi\) has a Kraus representation
\[
\Psi(M)=\sum_i X_iMX_i^\dagger,
\quad
\rank X_i\le2
\]
for all \(i\). Therefore
\[
(\Id_d\otimes\Psi)(C_\Phi)
=
\sum_i
(I_d\otimes X_i)C_\Phi(I_d\otimes X_i^\dagger).
\]
By the assumption, each summand is PPT, and hence
\((\Id_d\otimes\Psi)(C_\Phi)\) is PPT.
This proves \((iii)\Rightarrow(vi)\).

Conversely, assume that $(vi)$ holds. Let \(Y\in\M_d\) be a matrix with \(\rank Y\le2\).
Define the CP map $\Psi_Y$ as
\[
\Psi_Y(M)=YMY^\dagger.
\]
Then \(\Psi_Y\in\cSP_2\).
By the assumption,
\[
(\Id_d\otimes\Psi_Y)(C_\Phi)
=
(I_d\otimes Y)C_\Phi(I_d\otimes Y^\dagger)
\]
is PPT. 
Thus \((iii)\Leftrightarrow(vi)\), and the proof is complete.
\end{proof}

Finally, we verify the local filtering property used in Sec.~\ref{sec:swapping}.
The cones $\SEP$, $\sPPT$, $\sS_n$, and $\sUND_n$ are closed under local filtering.
That is, if $\rho$ belongs to one of these cones, then $(L\ox R)\rho(L^\dagger\ox R^\dagger)$
belongs to the same one for all local operators $L$ and $R$, whenever the resulting operator is nonzero.

For $\rho\in\sPPT$,
\[
[(L\ox R)\rho(L^\dagger\ox R^\dagger)]^\Gamma
=
(L\ox R^*)\rho^\Gamma(L^\dagger\ox R^T),
\]
which preserves positivity.

For $\rho\in\sS_n$, write $\rho=\sum_i \ket{\psi_i}\bra{\psi_i}$ with $\SR(\ket{\psi_i})\le n.$ 
Then
\[
(L\ox R)\rho(L^\dagger\ox R^\dagger)
=
\sum_i (L\ox R)\ket{\psi_i}\bra{\psi_i}(L^\dagger\ox R^\dagger),
\]
and
\[
\SR((L\ox R)\ket{\psi_i})\le \SR(\ket{\psi_i})\le n,
\]
which implies $(L\ox R)\rho(L^\dagger\ox R^\dagger)\in\sS_n$.

Lastly, for $\rho=\rho_{AB}\in\sUND_n$ and any vector  $\ket\psi\in(\C^d)^{\ox n}\ox(\C^d)^{\ox n}$ with $\SR(\ket{\psi})\le2$,
we have
\[
\begin{aligned}
&\bra{\psi}
\left[((L\otimes R)\rho_{AB}(L^\dagger\otimes R^\dagger))^{\otimes n}\right]^\Gamma
\ket{\psi}  \\
&=
\bra{\xi}(\rho_{AB}^{\otimes n})^\Gamma\ket{\xi}\ge0,
\end{aligned}
\]
where
\[
\ket{\xi}
=
\left(\bigotimes_{k=1}^n
L_{A_k}^\dagger\otimes R_{B_k}^T
\right)\ket{\psi},
\]
and since $\SR(\ket{\xi})\le\SR(\ket{\psi})\le2$, the above inequality follows from
$\rho_{AB}\in\sUND_n$.
Hence
\[
(L\otimes R)\rho_{AB}(L^\dagger\otimes R^\dagger)\in\sUND_n.
\]


\section{Proof of Theorem~\ref{thm:und-sp2-eb-composable}}
\label{appen:pfthm_und-sp2-eb-composable}
Let $\Phi\in\cUND_1$ and $\Psi\in\cSP_2$.
Since $C_\Psi\in \sS_2$, there is a decomposition
\[
C_\Psi=\sum_i\proj{\psi_i},
\quad \SR(\ket{\psi_i})\le2
\]
for every $i$. 
By Eq.~\eqref{eq:SR}, for each $i$, $\ket{\psi_i}=(X_i\otimes I_3)\ket{\phi_3^+}$ for some $X_i\in\mathsf{M}_3$ with $\rank(X_i)\le 2.$
Then
\[
\begin{aligned}
C_{\Phi\circ\Psi}
&=
(\Id_3\otimes\Phi)(C_\Psi)\\
&=
(\Id_3\otimes \Phi)\sum_i 
(X_i\otimes I_3)\proj{\phi_3^+}(X_i^\dagger\otimes I_3)\\
&=
\sum_i
(X_i\otimes I_3)C_\Phi(X_i^\dagger\otimes I_3).
\end{aligned}
\]

Since $C_\Phi$ is $1$-undistillable, Proposition~\ref{prop:compression} implies that each summand is PPT on a $2\ox 3$ support. 
Since PPT is equivalent to separability in $2\ox 3$ systems~\cite{peres1996separability,horodecki1996necessary}, every summand is separable.
Hence $C_{\Phi\circ\Psi}$ is separable.


For the opposite order, use the Kraus representation of $\Psi$.
Since $\Psi\in\cSP_2$, we may write
\[
\Psi(M)=\sum_i X_iMX_i^\dagger,
\quad
\rank(X_i)\le2
\quad
\text{for all }i .
\]
Then
\[
\begin{aligned}
C_{\Psi\circ\Phi}
&=
(\Id_3\otimes\Psi)(C_\Phi)\\
&=
\sum_i 
(I_3\otimes X_i)C_\Phi(I_3\otimes X_i^\dagger).
\end{aligned}
\]
Since $C_\Phi$ is $1$-undistillable, Proposition~\ref{prop:compression} implies that each summand is PPT.
Moreover, each summand is supported on a $3\otimes2$ subspace.
Since PPT is equivalent to separability in $3\otimes2$ systems~\cite{peres1996separability,horodecki1996necessary}, each summand is separable.
Therefore $C_{\Psi\circ\Phi}$ is separable.
\(\square\)


\section{Order of composition}\label{app:order}

In this section, we give a simple example showing that the order of composition cannot be ignored in general.
That is, EB-composability need not be symmetric for arbitrary channel classes.

Consider two-qubit channels \(\Phi,\Psi:\M_2\to \M_2\) defined by
\[
\Phi
\left(
\begin{bmatrix}
m_{00} & m_{01}\\
m_{10} & m_{11}
\end{bmatrix}
\right)
=
\begin{bmatrix}
m_{00}+\frac12 m_{11} & \frac{1}{\sqrt2}m_{01}\\
\frac{1}{\sqrt2}m_{10} & \frac12 m_{11}
\end{bmatrix}
\]
and
\[
\Psi
\left(
\begin{bmatrix}
m_{00} & m_{01}\\
m_{10} & m_{11}
\end{bmatrix}
\right)
=
\begin{bmatrix}
\frac15 m_{00} & \frac{1}{\sqrt5}m_{01}\\
\frac{1}{\sqrt5}m_{10} & \frac45 m_{00}+m_{11}
\end{bmatrix}.
\]
Then
\[
(\Phi\circ\Psi)
\left(
\begin{bmatrix}
m_{00} & m_{01}\\
m_{10} & m_{11}
\end{bmatrix}
\right)
=
\begin{bmatrix}
\frac35 m_{00}+\frac12 m_{11} &
\frac1{\sqrt{10}}m_{01}\\
\frac1{\sqrt{10}}m_{10} &
\frac25 m_{00}+\frac12 m_{11}
\end{bmatrix},
\]
and
\[
(\Psi\circ\Phi)
\left(
\begin{bmatrix}
m_{00} & m_{01}\\
m_{10} & m_{11}
\end{bmatrix}
\right)
=
\begin{bmatrix}
\frac15 m_{00}+\frac1{10} m_{11} &
\frac1{\sqrt{10}}m_{01}\\
\frac1{\sqrt{10}}m_{10} &
\frac45 m_{00}+\frac9{10} m_{11}
\end{bmatrix}.
\]
In the computational basis
$\{\ket{00},\ket{01},\ket{10},\ket{11}\}$, the corresponding Choi matrices are
\[
C_{\Phi\circ\Psi}
=
\begin{bmatrix}
\frac35 & 0 & 0 & \frac{1}{\sqrt{10}}\\
0 & \frac25 & 0 & 0\\
0 & 0 & \frac12 & 0\\
\frac{1}{\sqrt{10}} & 0 & 0 & \frac12
\end{bmatrix}
,\quad
C_{\Phi\circ\Psi}^\Gamma=
\begin{bmatrix}
\frac35 & 0 & 0 & 0\\
0 & \frac25 & \frac{1}{\sqrt{10}} & 0\\
0 & \frac{1}{\sqrt{10}} & \frac12 & 0\\
0 & 0 & 0 & \frac12
\end{bmatrix}
\]
and
\[
C_{\Psi\circ\Phi}
=
\begin{bmatrix}
\frac15 & 0 & 0 & \frac{1}{\sqrt{10}}\\
0 & \frac45 & 0 & 0\\
0 & 0 & \frac{1}{10} & 0\\
\frac{1}{\sqrt{10}} & 0 & 0 & \frac{9}{10}
\end{bmatrix}
,\quad
C_{\Psi\circ\Phi}^\Gamma=
\begin{bmatrix}
\frac15 & 0 & 0 & 0\\
0 & \frac45 & \frac{1}{\sqrt{10}} & 0\\
0 & \frac{1}{\sqrt{10}} & \frac{1}{10} & 0\\
0 & 0 & 0 & \frac{9}{10}
\end{bmatrix}
.
\]

The only nontrivial $2\times2$ blocks of $C_{\Phi\circ\Psi}^{\Gamma}$ and $C_{\Psi\circ\Phi}^{\Gamma}$ have determinants
$\frac25\cdot\frac12-\frac1{10}=\frac1{10}>0$ and
$\frac45\cdot\frac1{10}-\frac1{10}=-\frac1{50}<0$, respectively.
Thus $C_{\Phi\circ\Psi}$ is PPT, whereas $C_{\Psi\circ\Phi}$ is non-PPT.
Since these are two-qubit matrices, the first one is separable, while the second is entangled.
Therefore \(\Phi\circ\Psi\) is entanglement breaking but \(\Psi\circ\Phi\) is not.
This shows that, in general, the EB-composability of one order of composition does not imply the EB-composability of the opposite order.


\section{Proof of Theorem~\ref{thm:maximality-und1-sp2}}\label{appen:pfthm_maximality-und1-sp2}

The implications \((i)\Rightarrow(ii)\) and \((i)\Rightarrow(iii)\) are exactly Theorem~\ref{thm:und-sp2-eb-composable}.

Conversely, assume that \((\mathcal X,\cSP_2)\) is EB-composable.
Let \(\Phi\in\mathcal X\). Then for every \(\Psi\in\cSP_2\),
\[
C_{\Phi\circ\Psi}
=
(\Id_3\otimes\Phi)(C_\Psi)
\]
is separable, and hence PPT. 
As \(\Psi\) ranges over \(\cSP_2\), the
Choi matrix \(C_\Psi\) ranges over \(\sS_2\). 
Therefore
\[
(\Id_3\otimes\Phi)(\sigma)
\]
is PPT for every \(\sigma\in\mathcal \sS_2\). 
By
Proposition~\ref{prop:compression} $(v)$, we have
$\Phi\in\cUND_1$.
Hence \(\mathcal X\subseteq\cUND_1\).

Now assume that \((\cSP_2,\cX)\) is EB-composable.
Let \(\Phi\in\mathcal X\). 
Then for every \(\Psi\in\cSP_2\),
\[
C_{\Psi\circ\Phi}
=
(\Id_3\otimes\Psi)(C_\Phi)
\]
is separable, and hence PPT. 
By Proposition~\ref{prop:compression} $(vi)$, this implies $\Phi\in\cUND_1$.
Hence \(\mathcal X\subseteq\cUND_1\).

Combining the implications proves the equivalence.
\(\square\)


\section{Entanglement swapping}\label{appen:entanglement swapping}
In this section, we prove Theorem~\ref{thm:closed-under-filtering-reduction}. 
In fact, we prove a slightly stronger equivalence between several formulations of entanglement swapping for cones of states that are closed under local filtering.

Let $\sX$ and  $\sY$ be cones of bipartite positive semidefinite operators on $\C^d\ox\C^d$ and assume that both cones are closed under local filtering.
Throughout this section, the input operators are always taken to satisfy $\rho_{AB}\in \sX$ and $\sigma_{CD}\in \sY$.

\noindent
\textbf{(A) $\ket{\phi_d^+}$-measurement formulation.}
For all $\rho_{AB}\in \sX$ and $\sigma_{CD}\in \sY$,
\[
\Tr_{BC}
[
(\rho_{AB}\ox\sigma_{CD})
(I_A\ox \ket{\phi^+}\bra{\phi^+}_{BC}\ox I_D)
]
\in \SEP(A:D).
\]

\medskip

\noindent
\textbf{(B) Pure state measurement formulation.}
For all $\rho_{AB}\in \sX$ and $\sigma_{CD}\in \sY$, and for
every pure state $|\eta\rangle_{BC}$,
\[
\Tr_{BC}
[
(\rho_{AB}\ox\sigma_{CD})
(I_A\ox |\eta\rangle\langle\eta|_{BC}\ox I_D)
]
\in \SEP(A:D).
\]

\medskip

\noindent
\textbf{(C) Arbitrary measurement formulation.}
For all $\rho_{AB}\in \sX$ and $\sigma_{CD}\in \sY$, and for
every positive operator $M_{BC}\ge0$,
\[
\Tr_{BC}
[
(\rho_{AB}\ox\sigma_{CD})
(I_A\ox M_{BC}\ox I_D)
]
\in \SEP(A:D).
\]

\medskip

\noindent
\textbf{(D) Locally filtered formulation.}
For all $\rho_{AB}\in \sX$ and $\sigma_{CD}\in \sY$, every
positive operator $M_{BC}\ge0$, and all local operators $L_A,R_D\in \M_d$,
\[
(L_A\ox R_D)
\Tr_{BC}
[
(\rho_{AB}\ox\sigma_{CD})
(I_A\ox M_{BC}\ox I_D)
]
(L_A^\dagger\ox R_D^\dagger)
\in \SEP(A:D).
\]

\medskip

\noindent
\textbf{(E) Separable-CP-map formulation.}
For all $\rho_{AB}\in \sX$ and $\sigma_{CD}\in \sY$, and every
separable completely positive map
\[
\Lambda: \M_d\otimes (\M_d \otimes \M_d)\otimes\M_d
\longrightarrow
\M_d\otimes (\M_d \otimes \M_d)\otimes\M_d
\]
with respect to the tripartition $A|BC|D$, one has
\[
\Tr_{BC}[\Lambda(\rho_{AB}\ox\sigma_{CD})]
\in \SEP(A:D).
\]

The implications
\[
(E)\Rightarrow(D)\Rightarrow(C)\Rightarrow(B)\Rightarrow(A)
\]
are immediate from the definitions. 
Indeed, (A) is the special case of (B) with $|\eta\rangle=\ket{\phi^+}$, (B) is the rank-one special case of
(C), (C) is the case of (D) with $L_A=I_A$ and $R_D=I_D$, and (D) is a special case of (E).

\begin{Lem}\label{lem:ABCDE-swapping-equivalence}
Let $\sX$ and $\sY$ be cones of bipartite positive semidefinite matrices on $\C^d\ox\C^d$.  
Suppose that they are closed under local filtering.
Then the formulations (A)–(E) above are equivalent.
\end{Lem}

\begin{proof}

It remains to prove $(A)\Rightarrow(E)$.  
We first prove $(A)\Rightarrow(B)$.  
Let $|\eta\rangle_{BC}$ be arbitrary.
There exists $X\in M_d$ such that
\[
|\eta\rangle_{BC}=(X\otimes I_C)\ket{\phi^+}_{BC}.
\]
Thus
\[
|\eta\rangle\langle\eta|_{BC}
=
(X\otimes I_C)\ket{\phi^+}\bra{\phi^+}_{BC}(X^\dagger\otimes I_C).
\]
Using cyclicity inside the partial trace over $BC$, we get
\[
\begin{aligned}
&\Tr_{BC}
[
(\rho_{AB}\otimes\sigma_{CD})
(I_A\otimes|\eta\rangle\langle\eta|_{BC}\otimes I_D)
] \\
&\quad =
\Tr_{BC}
[
(\rho'_{AB}\otimes\sigma_{CD})
(I_A\otimes\ket{\phi^+}\bra{\phi^+}_{BC}\otimes I_D)
],
\end{aligned}
\]
where
\[
\rho'_{AB}
=
(I_A\otimes X^\dagger)\rho_{AB}(I_A\otimes X).
\]
Since $\sX$ is closed under local filtering, $\rho'_{AB}\in\sX$.
Therefore (A), applied to $\rho'_{AB}$ and $\sigma_{CD}$, implies that the rank-one outcome above is separable. 
Hence (B) holds.

Next, (B) implies (C).  Let $M_{BC}\ge0$. 
Write its spectral decomposition as
\[
M_{BC}=\sum_k\lambda_k|\eta_k\rangle\langle\eta_k|,
\quad \lambda_k\ge0.
\]
By linearity of the partial trace,
\[
\begin{aligned}
&\Tr_{BC}
[
(\rho_{AB}\otimes\sigma_{CD})
(I_A\otimes M_{BC}\otimes I_D)
] \\
&\quad =
\sum_k \lambda_k
\Tr_{BC}
[
(\rho_{AB}\otimes\sigma_{CD})
(I_A\otimes|\eta_k\rangle\langle\eta_k|_{BC}\otimes I_D)
].
\end{aligned}
\]
Each summand is separable by (B), and SEP is a convex cone.
Thus (C) follows.

Now (C) implies (D), because the outside operation
$X\mapsto (L_A\otimes R_D)X(L_A^\dagger\otimes R_D^\dagger)$ is a local filter on the final bipartite system $A:D$ and hence preserves separability.

Finally, (D) implies (E).  
Since $\Lambda$ is separable with respect to $A|BC|D$, it admits a Kraus representation
\[
\Lambda(X)
=
\sum_j
(P_j\otimes Q_j\otimes R_j)
X
(P_j^\dagger\otimes Q_j^\dagger\otimes R_j^\dagger),
\]
where $P_j$ acts on $A$, $Q_j$ acts on $BC$, and $R_j$ acts on $D$.
Therefore
\[
\begin{aligned}
&\Tr_{BC}[\Lambda(\rho_{AB}\otimes\sigma_{CD})] \\
&=
\sum_j
(P_j\otimes R_j)
\Tr_{BC}
[
(\rho_{AB}\otimes\sigma_{CD})
(I_A\otimes Q_j^\dagger Q_j\otimes I_D)
]
(P_j^\dagger\otimes R_j^\dagger).
\end{aligned}
\]
For each $j$, $Q_j^\dagger Q_j\ge0$ is a positive operator on $BC$.
Thus each summand is separable by (D). 
Hence (E) holds.
\end{proof}

Theorem~\ref{thm:closed-under-filtering-reduction} is precisely the implication $(A)\Rightarrow(C)$.
This proves the theorem.
\end{document}